\documentclass[12pt]{article}
\usepackage{makeidx}  
\usepackage{amsmath}
\usepackage[left=2.5cm,right=2.5cm,top=2.5cm,bottom=2.5cm]{geometry}
\usepackage{authblk}
\usepackage{amssymb}
\usepackage{natbib}

\DeclareMathOperator*{\Int}{\text{ \normalfont  Int}}
\usepackage{amsthm}
\newtheorem*{remark}{Remark}
\newtheorem{Proposition}{Proposition}

\usepackage{comment}
\usepackage{graphicx}
\usepackage{graphbox}
\usepackage{rotating}
\usepackage{multirow}
\usepackage{hyperref}

\newcommand{\norm}[1]{\left\lVert#1\right\rVert}

\begin{document}
%
\date{}
\title{A Functional Data Framework For Analyzing Shapes and Textures in Images}
\author{Issam-Ali Moindjié}
\affil{\small LAMPS, Université de Perpignan Via Domitia, France\\
issam-ali.moindjie@univ-perp.fr}
\maketitle     
\begin{abstract}
    Images represent objects characterized by contours and textures. From a statistical perspective, these features can be defined as observations of continuous random functions. However, most existing approaches rely on pixel‑based discretizations, which lead to high‑dimensional representations and heavy computational costs. In this note, we introduce an alternative, more frugal representation. This representation assumes that the object has a \textit{star-shaped domain} interior. Under this condition, we explore the analysis of images from a functional data analysis perspective. The proposed framework is illustrated on a real data supervised image classification problem.
\end{abstract}
\section{Introduction}
With the advances of data acquisition, images have been extensively collected in many areas with various purposes. In medicine, they provide useful guidance for diagnosis-making (see e.g \cite{kuhl2007}, \cite{goyal2020}), and in agronomy sciences, images can be used to evaluate fruits' ripeness for food quality control (\cite{rizzo}). Although images are increasingly available, statistical analysis of images is still challenging. Particularly under the pixel-based representation which renders classical statistical methods unsuitable as images are seen as high-dimensional matrices with complex correlation structures. 
\par 
Therefore, new methods have been proposed to address the challenges of image analysis. The most popular one is the convolutional neural networks method \citep{LeCun1989}, which has achieved remarkable performance in supervised image classification problems. However, it suffers from some well-known drawbacks: it is  uninterpretable \citep{zhang2018} and its estimation is computationally demanding \citep{deep-comp}. These problems might worsen with the advances in acquisition technologies, as they will give more detailed images, which, under the pixel-based representation, will lead to more complex matrices to analyze. Consequently, the accessibility to CNN and other deep learning models will be progressively limited. In this context, there is a need for the development of frugal statistical methods for analyzing images. \par 
This note aims to contribute to this effort by introducing a new framework for image analysis. It relies on statistical shape analysis  \citep{dryden2016} and functional data analysis \citep{ramsay2005} to account for the continuous nature of the main aspects of object represented in images:  its contour and its interior color variation (texture). \par 
The contour of the object is seen as the \textit{shape} of the objects (in the sense of \cite{kendall}) plus the action of deformations variables such as scaling, rotation, etc. Formally, the contour is a closed planar curve $\mathbf{C}: [0, 1]\to \mathbb{R}^2$, where $t\in [0, 1]$ represents the proportion of the curve that has been traveled from the start ($t=0$) to the end ($t=1$). As, $\mathbf{C}$ is closed, $\mathbf{C}(0)= \mathbf{C}(1)$, with $\mathbf{C}(0)$ and $\mathbf{C}(1)$ are respectively the coordinates at beginning and the end of the curve. Statistical analysis of contours is related to the shape analysis and has been subject of several works. However, historically, the contributions have been focused on annotated landmarks, i.e. sampled finite contours, where the correspondence between all observational points is known for each curve (see e.g. \cite{dryden2016}). This hypothesis is limited in image analysis, since annotating landmarks is highly subjective and time-consuming. Modern works have introduced the random planar curves setting \citep{srivastava2016}, but rely on discrete observations for the estimation process. In this note, we follow the recent works of \cite{axe1}, which presents a functional data analysis approach for studying the contours.\par
Previous works have considered texture only on subsampled contours. In \cite{cootes}, the authors propose to integrate colors for a joint analysis of shapes and appearances (color variations inside and outside the objects). However, this approach is dependent to have a very precise set of landmarks for each image, which is of course non-translatable in most situations. \par 
This note proposes an extension of the work of \cite{axe1}, by incorporating textures into the analysis for a general class of contours: objects with star-shaped domains. This property makes it possible to study color variation inside the objects, and it is held for several objects encountered in image analysis (e.g. biological cells, leaves, organs, etc.). We say that the closure of the interior of a represented object in image ($\mathcal{D}\subset \mathbb{R}^2$) is a star-shaped domain at the point $\mathbf{c}\in \mathbb{R}^2$, if 
\begin{equation} u \in \mathcal{D} \implies 
    \left\{ (1-\lambda)\mathbf{c} + \lambda u\mid \lambda \in\left[0,1\right]\right\} \subset \mathcal{D}.
    \label{star}
\end{equation}
Figure~\ref{conv} illustrates this idea. The bottle shown in Figure~\ref{conv}(a) has a star-shaped interior, yet it is clearly non-convex. This shows that the star-shaped propriety is a weaker condition than convexity, as also expressed in~\eqref{star}. Indeed, every convex set is star-shaped with respect to \emph{any} point in the set, not just a particular center~$\mathbf{c}$. Consequently, convex sets form a subclass of star-shaped domains, whereas the converse does not generally hold.
\par 
\begin{figure}[h]
    \centering
    \begin{tabular}{c c c c c c c  }
    (a) 
        &\includegraphics[scale=0.2, align=c]{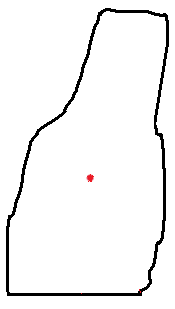}    & 
        (b) &        
          \includegraphics[scale=0.2, align=c]{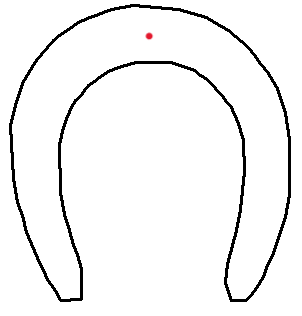} &  \\   (c) &\includegraphics[scale=0.2, align=c]{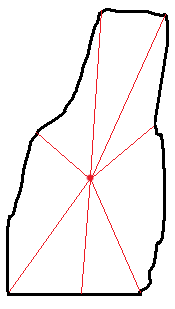} & (d) & \includegraphics[scale=0.2, align=c]{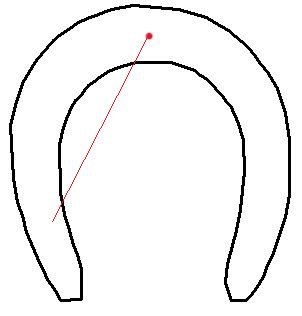}
    \end{tabular}
    \caption{Examples of two objects, bottle with a star-shaped domain interior (a\& c), and horseshoe with a non-star-shaped domain interior(b\& d); The center $\mathbf{c}$ (see Section \ref{conts} for details) for each figure is represented in red (top row figures: a, b). The bottle has a star-shaped domain, contrary to the horseshoe, as (d) shows: there exists lines which fall outside its interior. Images are drawn using the curves provided in the MPEG-7 datasets (\url{https://dabi.temple.edu/external/shape/MPEG7/dataset.html}) }
    \label{conv}
\end{figure}
In this work, we demonstrate the appealing proprieties of star-shaped propriety by explicitly studying the underlying mapping. Then, we introduce a joint functional data analysis of contours and textures. 
\par The manuscript proceeds as follows. After recalling main concepts in Section~\ref{back}, we study the new representation of colors in images in Section~\ref{text-sec} and Section~\ref{stat}. Section~\ref{app} presents an application to real data for supervised image classification using this framework, and Section \ref{disc} concludes this note, by giving perspectives of future works.
\section{Background}
\label{back}
In this work, we study the image $\mathbf{I}$, which we assume is a random surface taking values in $\mathcal{I}=L_2([0, 1]^2, \mathbb{R})$: the set of squared integrable functions defined from $[0,1]^2$ to $\mathbb{R}$.  
The space $\mathcal{I}$ is a Hilbert space endowed with the classical inner product: 
$$
 \langle \boldsymbol{f}, \boldsymbol{g} \rangle_{\mathcal{I}}= \int_{0}^1  \int_{0}^1\boldsymbol{f}\begin{pmatrix}
     x\\
     y
 \end{pmatrix}\boldsymbol{g}\begin{pmatrix}
     x\\
     y
 \end{pmatrix}dxdy,\ 
 $$
for $\boldsymbol{f}, \boldsymbol{g} \in \mathcal{I}$. For the sake of notation, the definition of $\mathcal{I}$ assumes that $\mathbf{I}$ is a real random surface, i.e. the image $\mathbf{I}$ is in grayscale. However, the generalization of our work to Red–Green–Blue (RGB) image is direct, as the application part shows. 

Although $\mathbf{I}$ is a random surface, technological limitations prevent its continuous observation. In practice, each image $\mathbf{I}_i$ (for $i = 1,\ldots,n$) is recorded on a discrete grid that reflects its resolution. Since the resolution may vary from one image to another, each $\mathbf{I}_i$ can be observed on its own grid $G_i \subset [0,1]^2$:
$$
\mathbf{I}_i(t_{i,j}), \qquad t_{i,j} \in G_i.
$$
The points $t_{i,1}, \ldots, t_{i,|G_i|}$ correspond to the pixels, with $|G_i|$ denoting the number of pixels in the grid. In this pixel-based representation, models are built from the collection of observations $\{\mathbf{I}_i(t_{i,j}),\ t_{i,j}\in G_i\}_{i=1,\ldots,n}$. However, when the grids are large, the computational cost becomes substantial for estimating models.

This work tackles this classical representation by explicitly considering the main aspect of represented objects in the images: shapes, and the textures. The next section provides a central step for obtaining these variables.

\subsection{Image segmentation}
The image $\mathbf{I}$ represents a scene, potentially with several objects. Our approach proposes to explicitly include these objects in the analysis. Therefore, a natural first step is to identify these objects. This is the segmentation process- an active research area of computer vision (see e.g. \cite{fu1981survey}, \cite{zaitoun2015}, \cite{otsu}). \par This note focuses on the case where the image $\mathbf{I}$ contains a single object. The objective is then to obtain $\mathbf{I}'$, the silhouette of this object:
$$
\mathbf{I}'(u) = \left\{\begin{array}{c c}
1& u \in \mathcal{D} \\
0& \text{otherwise}
\end{array} \right.
$$ 
with $u\in [0, 1]^2$, and $\mathcal{D}\subset [0, 1]^2$ represents the closure of the interior of the object, i.e. frontier plus the interior. \par 
In our application examples, we use a classical procedure, originally proposed in \cite{otsu}. It optimally selects a threshold to segment the image based on the pixels' intensity values. Figure~\ref{segm} presents the obtained results on a random example.  
\begin{figure}[ht]
    \centering
    \begin{tabular}{c c}
    (a) & (b) \\
          \includegraphics[width=0.25\linewidth]{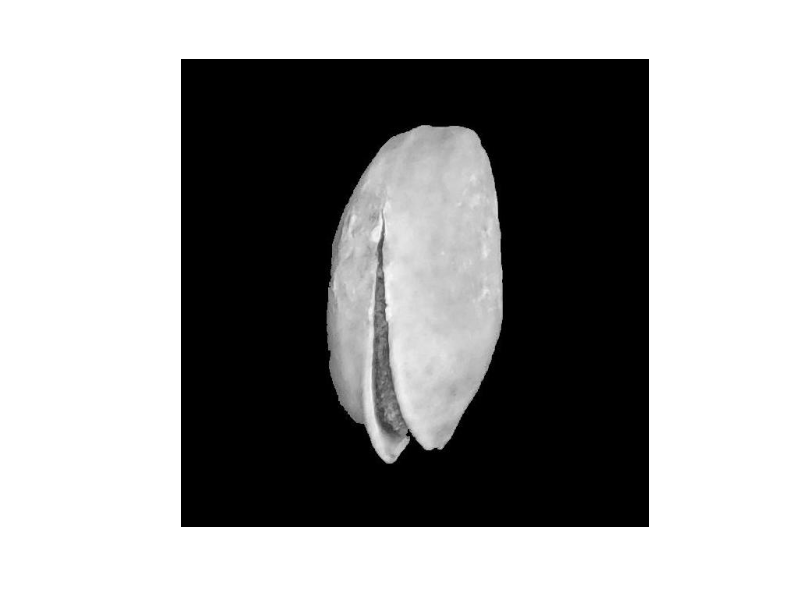}   &     \includegraphics[width=0.25\linewidth]{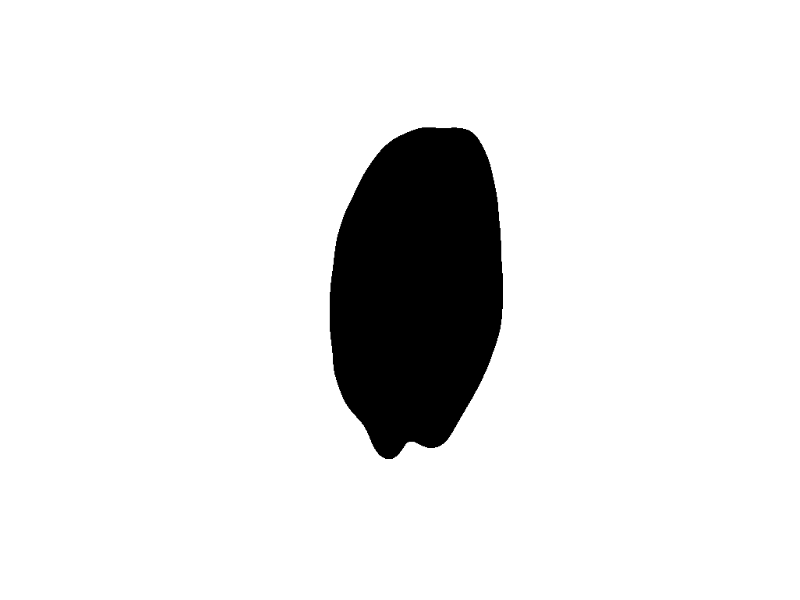}\\
    \end{tabular}
    \caption{Segmentation using  \cite{otsu}: (a) the original image (in gray-scale) and (b) the estimated silhouette. The original image is from dataset provided in \cite{singh2022}. }
    \label{segm}
\end{figure}
\subsection{Contours }
\label{conts}
\par The binary surface $\mathbf{I}'$ represents the \textit{silhouette} of the object in $\mathbf{I}$. It encodes the \textit{shape} of the object. While one can suggest using $\mathbf{I}'$ to consider the shape of the object, this representation is suboptimal (see \cite{srivastava2016} for a discussion). For the further analysis, inspired by \cite{axe1}, we focus on the contour $\mathbf{C}$ of the object, as it is a more parsimonious representation of the object's form. The contour variable $\mathbf{C}$ is defined as $\mathbf{C} : [0, 1]\to \mathbb{R}^2$, where  $$
\mathbf{C}(t)= \begin{pmatrix}
    C_x(t) \\ 
    C_y(t)
\end{pmatrix},
$$
and  $t\in [0, 1]$ represents the proportion of the curve that has been traveled from the start ($t=0$) to the end ($t=1$) and $C_x$, $C_y$ are the coordinate functions. Formally, the contour is a simple planar curve: it is closed, $\mathbf{C}(0)=\mathbf{C}(1)$ and $\mathbf{C}$ is injective on $[0, 1) \to \mathbb{R}^2$. \par
Like $\mathbf{I}'$, $\mathbf{C}$ encodes the \textit{silihouette} of the object in $\mathbf{I}$, by delimiting the interior and the exterior of the object. In practical cases, we estimate contour using the marching squares algorithm on the silhouette image \citep{maple2003}, see Figure \ref{cont} for an illustration. 
\begin{figure}[ht]
    \centering
    \begin{tabular}{c c c }
        (a) & (b) & (c) \\ 
        \includegraphics[width=0.3\linewidth, align=c]{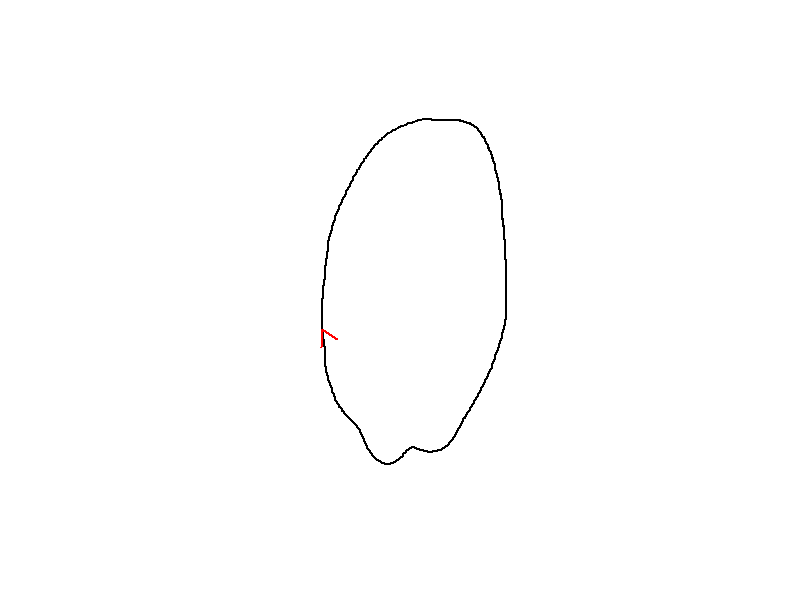} &
                    \includegraphics[width=0.20\linewidth, align=c ]{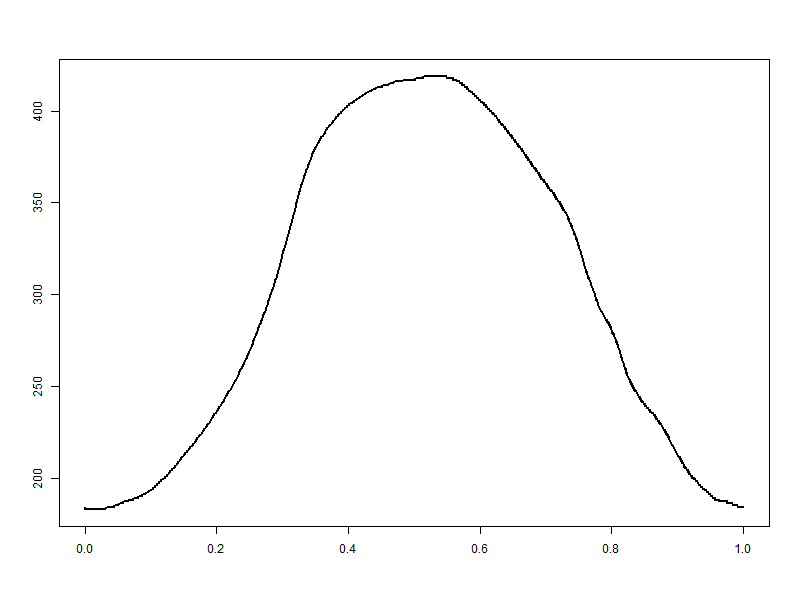}& 
                                \includegraphics[width=0.20\linewidth, align=c ]{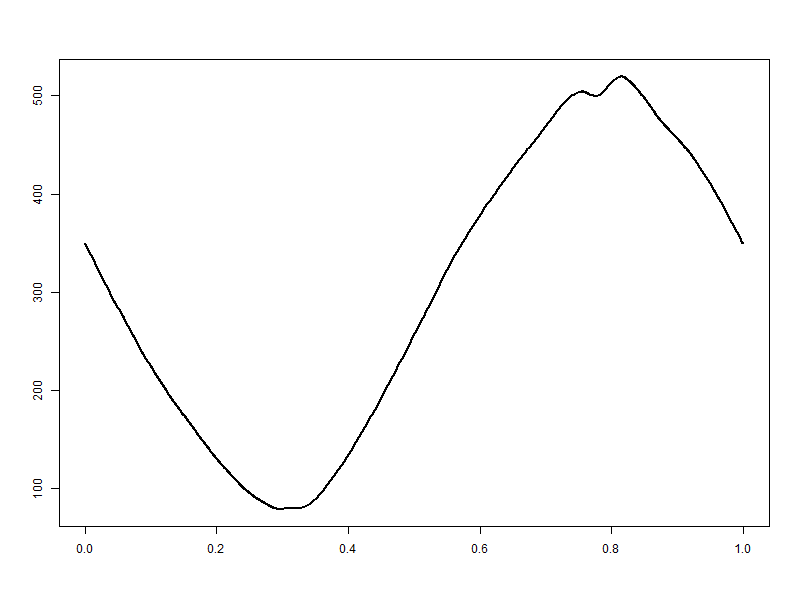}
    \end{tabular}
    \caption{Estimation of the contour using the marching square algorithm \cite{maple2003}: (a) is the estimated contour from $\mathbf{I}'$ (Fig~\ref{segm}), where the red arrow indicates the starting point ($t=0$) and (b)\&(c) are the associated coordinate functions: $C_x(t)$ and $C_y(t)$. }
    \label{cont}
\end{figure}
\par For the statistical analysis of the contours, we follow the recent work of \cite{axe1}, which assumes that $\mathbf{C}$ is random planar curve taking values in the Hilbert space $\mathcal{C}= L_2([0,1], \mathbb{R}) \times L_2([0,1], \mathbb{R})$ holding :
$$
\mathbf{C}=\rho \mathbf{O}\mathbf{\tilde C} \circ \gamma_\delta + \mathbf{c}
$$
where $(\rho, \mathbf{O}, \mathbf{c}, \gamma_\delta)$ are the deformation variables and $\mathbf{\tilde C}$ is the \textit{shape} of $\mathbf{C}$ in the sense of \cite{kendall}, i.e, what remains when all deformations are discarded. The shape variable takes value in the quotient space $\mathbf{S}_\infty\symbol{92} \sim$ where $$
\mathbf{S}_\infty
= \left\{\boldsymbol{f}=\begin{pmatrix}
    f_x & f_y
\end{pmatrix}^\top \in \mathcal{C}, \ \int_0^1 f_x(t)dt=\int_0^1 f_y(t)dt=0,\ \norm{\boldsymbol{f} }_\mathcal{C}=1 \right\},$$ 
and $\boldsymbol{f}, \boldsymbol{g} \in \mathbf{S}_\infty$ hold $\boldsymbol{f}\sim\boldsymbol{g}$ when they share the same shape, i.e. $\boldsymbol{\tilde f}= \boldsymbol{\tilde g}.$\par 
The deformation variables transform $\mathbf{\tilde C}$ to $\mathbf{C}$ by the action of uniform scaling ($\rho\in \mathbb{R}^+$), translating ($\mathbf{c}\in \mathbb{R}^2$), rotating ($\mathbf{O}\in SO(2)$) and reparametrizing ($\gamma \in \Gamma$, where $\Gamma$ is the space of reparametrization functions from $[0,1]\to [0,1]$). 
The translation and scaling deformations can be estimated using the definition of $\mathbf{\tilde C}$ : $$ \displaystyle \mathbf{c}=\begin{pmatrix}
    \int_0^1 C_x(t)dt & 
    \int_0^1 C_y(t)dt
\end{pmatrix}^\top  \text{ and} \ \rho=\norm{\mathbf{C}-\mathbf{c}}_{\mathcal{C}}.$$  For estimating the remaining parameters, we need first to determine the Fréchet mean $\boldsymbol{\mu}$.  Then, using $\boldsymbol{\mu}$, we estimate the rotation and the reparametrization deformations by the Iterative Closest Function algorithm, as suggested in \cite{axe1}. \par In the following, we will assume that the whole set of deformation variables is estimated which allows to  define the aligned and centered  planar curve $\mathbf{S}$ as $$
\mathbf{S}(t)=\rho \mathbf{\tilde C}(t), \ 
$$
with $t\in [0,1]$. Remark that by definition, $\mathbf{S}=\begin{pmatrix}
    S_x & S_y
\end{pmatrix}^\top$ holds $\norm{\mathbf{S}}_{\mathcal{C}}= \rho$ and $\int_0^1S_x(t)dt=\int_0^1S_y(t)dt=0$. \par Using $\mathbf{S}$, the next section focuses on the analysis of texture under the star-shaped domain hypothesis.

\section{Textures}
\label{text-sec}
In \cite{cootes}, the authors analyze images using  training
set in which corresponding ‘landmark’ points have been marked on every image. Under this framework, their approach  studies discrete textures and discrete contours (or landmarks) in images. However, since, in most real cases, images are not provided with meaningful landmarks, their approach is limited. This section proposes a less restrictive framework. \par 
We assume that we have the contours $\mathbf{C}$ associated with the image $\mathbf{I}$, and we derive the aligned and centered  curve $\mathbf{S}$. Based on them, the aim of this section is to study the texture, which we define as the function $\mathbf{T}^*$: 
$$
\mathbf{T}^*(u):=\mathbf{I}(u)
$$
where $u\in \mathcal{D}$ and $\mathcal{D}$ is the set of points inside the object. Note that unlike the approach of  \cite{cootes}, our definition of the texture $\mathbf{T}^*$ only considers the color variations inside the object. Moreover, it is clear that this definition heavily relies on the set $\mathcal{D}$: the domain of $\mathbf{T}^*$. As different objects may have different sets $\mathcal{D}$, it is necessary to define a global parametrization of $\mathcal{D}$. This is the aim of the next result, which gives a global parametrization of $\mathbf{T}^*$, under the hypothesis of $\mathcal{D}$ being a star-shaped domain. 
\begin{Proposition}If $\mathcal{D}$ is a star-shaped domain at \textit{center} $\mathbf{c}$, then for $v\in \mathbf{D}=\{ z\in \mathbb{C}, |z|\leq 1\}$  
$$
  \mathbf{O}\varphi(v)+\mathbf{c}\in \mathcal{D},
$$
with   
$$
\varphi(v)= |v|\mathbf{S}\left(\frac{\text{arg}(v)}{2\pi}\right).
$$ 
\label{prop1}
\end{Proposition}
The detailed proof of Proposition \ref{prop1} is  given in Appendix \ref{annex}, for the sake of readability. This result shows that there exists a mapping of $\mathcal{D}$ with the unit disk $\mathbf{D}$, which allows the definition of textures on a common domain $\mathbf{D}$, instead of $\mathcal{D}$: \begin{equation}
    \mathbf{T}(v)=\mathbf{I}(\mathbf{O}\varphi(v)+\mathbf{c})\ 
\label{text_eq}
\end{equation}
where $ v\in \mathbf{D}$. Note that $\mathbf{T}(v)$ is a real-random surface defined on the complex unit disk $\mathbf{D}$. In the following, for the sake of simplicity in the calculations and illustrations, we use the following bijective transformation from $[0, 1]^2 \to \mathbf{D} $, $
g(x, y)= xe^{i2\pi y}, 
$ such as $\mathbf{T}\circ g$ takes values in $\mathcal{I}$. Figure \ref{text-fig} presents the texture $\mathbf{T}\circ g$ from our previous example. Figure \ref{text-fig} (b) shows that the texture representation is independent of the object’s contour.
\begin{figure}[ht]
\centering
\begin{tabular}{ c c }
(a) & (b) \\
 \includegraphics[width=0.25\linewidth]{Figs/ex/kirmizi.png}     &  \includegraphics[width=0.25\linewidth]{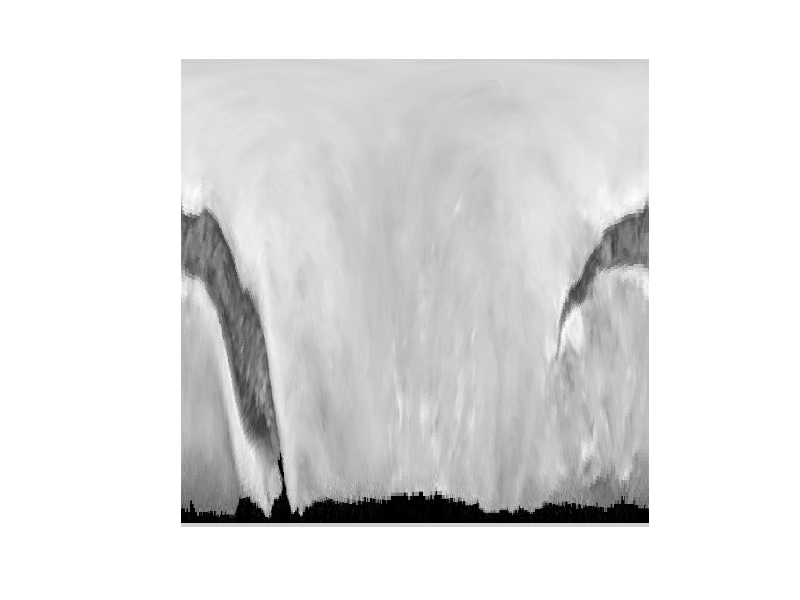}    
\end{tabular}
 \caption{Example of the texture, (a) image in gray-scale and (b) the associated texture (using $g$). }
 \label{text-fig}
\end{figure}
\par Proposition \ref{prop1} allows to map $\mathcal{D}$ to $\mathbf{D}$, which has several practical advantages. However, if the statistical analyses are performed using this transformation, it is important to study the inverse mapping, as it will be more convenient for interpretation. 
\begin{Proposition} Let $\alpha$ be the \emph{angle function} associated with $\mathbf{S}$: 
$$ \alpha(\theta)= \text{arg}\left(\begin{pmatrix}
    1 & i
\end{pmatrix} \mathbf{S}\left(\frac{\theta}{2\pi}\right) \right).
$$ 
If $\alpha(\theta)$ is injective on $(0, 2\pi] \to (0, 2\pi]$, and $u \in \mathcal{D}/\{\mathbf{c}\}$ , there exists a unique $v\in \mathbf{D}$,  such as $$
u=\mathbf{O}\varphi(v)+\mathbf{c},
$$
where $v=\lambda e^{i\theta^*}$, with 
$\alpha(\theta^* )= \arg\left( \begin{pmatrix}
    1 & i
\end{pmatrix} (u-\mathbf{c}) \right)$, $\lambda=  \norm{u-\mathbf{c}}_2/\norm{\mathbf{S}\left(\frac{\theta^*}{2\pi}\right) }_2 $ and $\norm{\cdot}_2$ is the Euclidean $2$-norm of $\mathbb{R}^2$. 
\label{prop2}
\end{Proposition}
We give a detailed proof of Proposition \ref{prop2} in Appendix \ref{annex}. 
This result says that for all elements $u \in \mathcal{D}/\{\mathbf{c}\}$, the inverse mapping is obtained using $$
\mathbf{I}(u)= \mathbf{T}\left(\frac{\norm{u-\mathbf{c}}_2}{\norm{\mathbf{S}\left(\frac{\theta^*}{2\pi}\right) }_2} \varphi(e^{i\theta^*})+ \mathbf{c} \right),
$$
where $
    \alpha(\theta^*)= \arg\left( \begin{pmatrix}
    1 & i
\end{pmatrix} (u-\mathbf{c}) \right),$  and if $u=\mathbf{c}$, we have that 
$\mathbf{I}(u)= \mathbf{T}\left(\varphi(0_\mathbb{C})+ \mathbf{c} \right)$, meaning that $\mathbf{c}$ corresponds to the 
center of $\mathbf{D}$. Hence, obtaining the inverse transformation of $u\in \mathcal{D}/\{\mathbf{c} \}$ reduces to resolve numerically the problem $$
    \alpha(\theta^*)= \theta,  
\text{where } u= \norm{u-\mathbf{c}}_2\begin{pmatrix}
    \cos(\theta) \\ 
    \sin(\theta)
\end{pmatrix}+\mathbf{c}.$$
\section{Statistical analysis of textures and contours}
\label{stat}
As images are observed in a discrete grid, a crucial step, before the analysis, is to estimate the functional form of $\mathbf{T}$ and $\mathbf{C}$. In this work, for doing so, we use the basis expansion technique, a common functional data analysis technique \citep{ramsay2005}. It assumes that the texture $\mathbf{T}$ and the contour $\mathbf{C}$ can be sufficiently well-approximated in a finite set of functions: 
 $$
\mathbf{T}\circ g(u)= \sum_{k=1}^{M_1}a_k\phi_k(u), \ \ \text{and }  \mathbf{C}(t)=\mathbf{c}+\sum_{k=1}^{M_2} \mathbf{c}_k \psi_k(t)
$$
where $u\in [0, 1]^2$ and $t\in [0, 1]$, $\{\phi_k\}_{k=1, \ldots, M_1}$ and $\{\psi_k\}_{k=1, \ldots, M_2}$
are some chosen basis. In these equations, the random variables $a_1, \ldots, a_{M_1}$ are real valued and $\mathbf{c}_1, \ldots, \mathbf{c}_{M_2}$ takes values in $\mathbb{R}^2$. 
\par Several bases can be chosen depending on the nature of the functional variables (see \cite{ramsay2005} for details). Here, we define $\{\phi_k\}_{k=1, \ldots, M_1}$ as $2$-d splines to have a maximum of details with low number of basis functions when approximating the texture function forms, and following \cite{axe1}, we define $\{\psi_k\}_k$ as the basis of Fourier functions. Under this hypothesis, we can obtain the shape $\mathbf{\tilde C}$ from the contour, by  
$$
\mathbf{\tilde C}= \frac{1}{\rho} \mathbf{O}^\top  \begin{pmatrix}
    \mathbf{c}_1 & \ldots & \mathbf{c}_{M_2} 
\end{pmatrix}  \boldsymbol{\beta}(\delta) \begin{pmatrix}
    \psi_1(t) \\
    \vdots \\
    \psi_{M_2}(t)
\end{pmatrix}
$$
where $\boldsymbol{\beta}(\delta)$ is an $M_2\times M_2$ orthogonal matrix which accounts for the reparametrization deformations, see \cite{axe1} for details. \par 
The statistical analysis of images is then performed by explicitly considering the two main components of the image $\mathbf{I}$: the texture and the shape of the object, $\mathbf{\tilde C}$. For doing so, we define the variable $\mathbf{Z}$, $$
\mathbf{Z}(z)= \begin{pmatrix}
    \mathbf{T}\circ g(u)&  
    \mathbf{\tilde C}^\top(t) 
\end{pmatrix}^\top , 
$$
where  $z=(u, t)$. Then, $\mathbf{Z}$ is a functional variable defined on different domains as it is valued in $\mathcal{I}\times \mathbf{S}_\infty$ \citep{happ}. 
\subsection{Classification of images using textures and contours} 
To show the applicability of our framework in statistical learning problems, this section presents its adaption for supervised classification. We consider that a binary response variable $Y$ is associated with the image $\mathbf{I}$ (and so with $\textbf{Z}$), and the aim is to estimate through a statistical model the relationship of $\mathbf{Z}$ with $Y$. Our focus is on interpretable models and therefore our attention is restricted to linear relationships.
\par As the shape $\mathbf{ \tilde C}$ takes values in a subspace of $\mathbf{S}_\infty$, which is a non-linear manifold with a complex geometry, it is necessary to use a projection of $\mathbf{\tilde C}$ into a more convenient linear space. In this work, we use the projection of $\mathbf{\tilde C}$ to the tangent space on at the point $\boldsymbol{\mu}$, its Fréchet mean. For a given shape $\mathbf{\tilde C}$ this projection is defined as 
$$
\mathcal{T}_{\boldsymbol{\mu}}(\mathbf{ \tilde C})= \mathbf{\tilde C}-  \boldsymbol{\mu}\langle \mathbf{\tilde C}, \boldsymbol{\mu}\rangle_\mathcal{C} .  
$$ 
Hence, linear models can be written as $$
\Psi(Y)= \beta_0+ \langle \boldsymbol{\beta}_{\mathbf{T}}, \mathbf{T}\circ g \rangle_{\mathcal{I}}+ \langle  \boldsymbol{\beta}_{\mathbf{\tilde C}} , \mathcal{T}_{\boldsymbol{\mu}}(\mathbf{\tilde C}\rangle_{\mathcal{C}}+ \epsilon,
$$
where $\epsilon$ is a random noise uncorrelated to $Y$, $\boldsymbol{\beta}= \begin{pmatrix}
    \boldsymbol{\beta}_{\mathbf{T}} \\\boldsymbol{\beta}_{\mathbf{\tilde C}}
\end{pmatrix}$ are the regression coefficients and
\begin{equation}
\Psi(Y) = \left\{
    \begin{array}{ll}
     \sqrt{\frac{\pi_1}{\pi_0}}     &\text{, if }  Y=0\\
        - \sqrt{\frac{\pi_0}{\pi_1}} &\text{, if } Y=1, 
    \end{array}
\right.
\label{reco}
\end{equation}
with $\pi_0= \mathbb{P}(Y=0)$ and $\pi_1=1-\pi_0$. This transformation allows estimating the linear discriminant coefficient $\boldsymbol{\beta}$ using the  regression of $\Psi(Y)$ on $\mathbf{Z}^*$, see \cite{preda}. However, as the regression with functional covariates is an ill-posed problem, we use the following methods for estimating $\boldsymbol{\beta}$: Principal components for multivariate functional data (MFPCA, \cite{happ}) and Partial least square regression for multivariate functional data (MFPLS, \cite{pls}). Both procedures address the problem of estimating $\boldsymbol{\beta}$ by projecting the predictor variable $\mathbf{Z}^*= \begin{pmatrix}
    \mathbf{T}\circ g&  
    \mathcal{T}_{\boldsymbol{\mu}}(\mathbf{\tilde C})^\top
\end{pmatrix}^\top $ into a convenient set of functions $\{\boldsymbol{f}_k\}_{k=1}^M:$ 
$$
\mathbf{Z}^*-\mathbb{E}(\mathbf{Z}^*) \simeq \sum_{k=1}^{M} \xi_k \boldsymbol{f}_k. 
$$
Here $\xi_k$ are obtained as $\xi_k= \langle \mathbf{Z}^*-\mathbb{E}(\mathbf{Z}^*), \boldsymbol{f}_k\rangle_{ \mathcal{I}\times \mathcal{C}}$ for $k=1, \ldots, M$. However, MFPLS and MFPCA differ in the used criteria for obtaining $\boldsymbol{f}$; on one hand, MCPCA uses the covariance of $\mathbf{Z}^*$, i.e. the representation seeks to optimally reconstruct the variabilities of $\mathbf{Z}^*$. On the other hand, MFPLS tries to obtain the scores which capture the main variations of the covariance operator between $\mathbf{Z}^*$ and $Y$.
\section{Application}
\label{app}
This section presents an application of the proposed method on real data. We use the datasets provided in \cite{singh2022}, which are composed of $2148$ images of $600 \times 600 $ resolution representing two different types of pistachios cultivated in Turkey: \textit{Kirmizi} and \textit{Siirt}. They contain $1232$ instances of \textit{Kirmizi} and $916$ instances of \textit{Siirt}. Figure \ref{xx} presents several examples. It shows that these images have been preprocessed to remove the background and to uniformize the size of pistachios. However, some deformations variables such as rotations remain present in images.\par 
\begin{figure}[h]
    \centering
    \begin{tabular}{c c c | c c}
    & \multicolumn{2}{c}{(\it Kirmizi)} & \multicolumn{2}{c}{(\it Siirt)} \\ 
         &\includegraphics[align=c, width=0.15\linewidth]{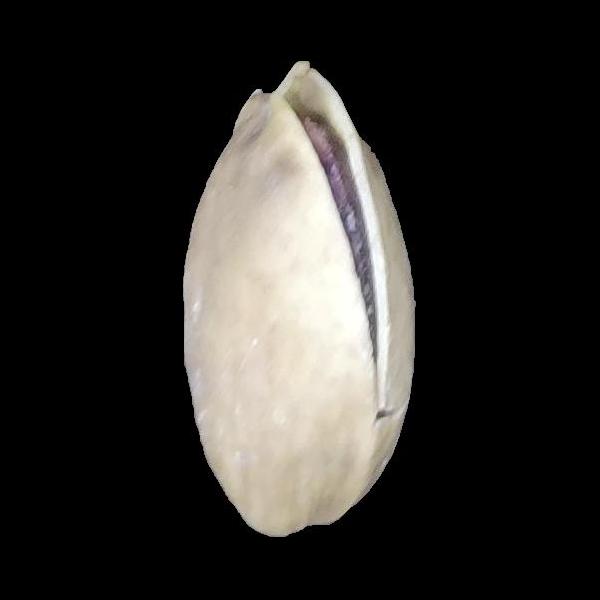}     & \includegraphics[align=c, width=0.15\linewidth]{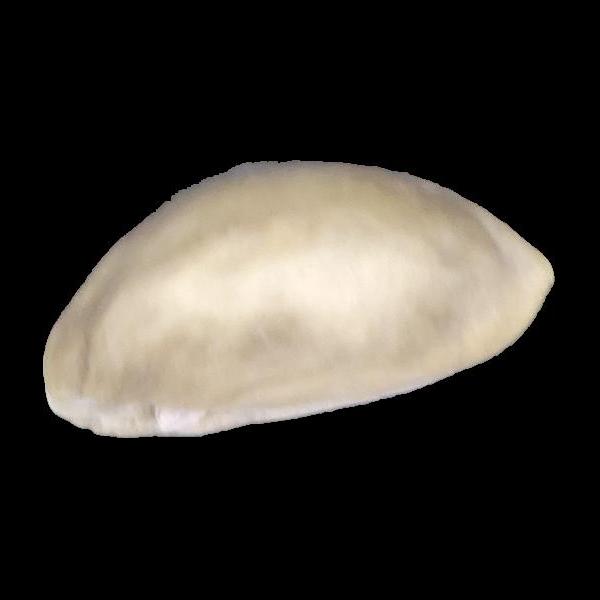} &
         \includegraphics[align=c, width=0.15\linewidth]{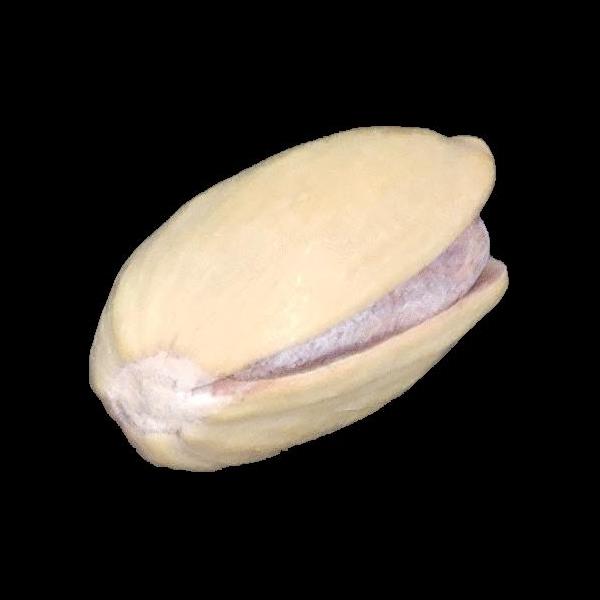}     & \includegraphics[align=c, width=0.15\linewidth]{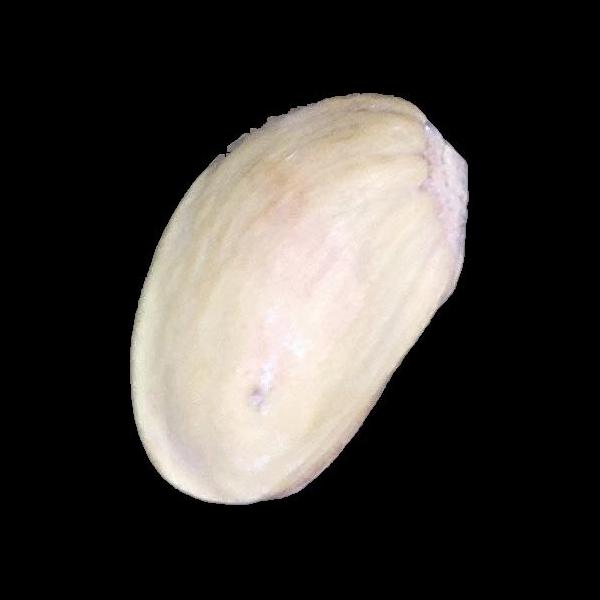} \\
    \end{tabular}
    \caption{Examples of two \textit{Kirmizi} and \textit{Siirt} pistachios \citep{singh2022}. Each image has a resolution of $600\times 600$.   }
    \label{xx}
\end{figure}
This numerical study aims to evaluate if the proposed models are able to recognize each species of pistachios even if they are not aligned. For doing so, as the images have been globally standardized, we add some random rotating and zooming  effects to assess the robustness of the proposed framework. Specifically, the rotation angle was drawn using uniform distribution in $[0, 2\pi]$ and the zooming scale was drawn from univariate distribution in $[0.2, 1]$. Figure \ref{xy} presents the obtained outputs using the images of Figure \ref{xx}. 
\begin{figure}[h]
    \centering
    \begin{tabular}{c c c | c c}
    & \multicolumn{2}{c}{(\it Kirmizi)} & \multicolumn{2}{c}{(\it Siirt)} \\ 
         &\includegraphics[align=c, width=0.15\linewidth]{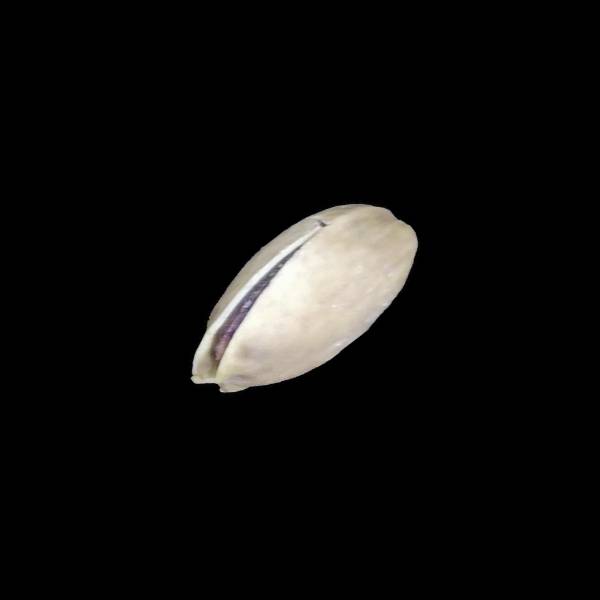}     & \includegraphics[align=c, width=0.15\linewidth]{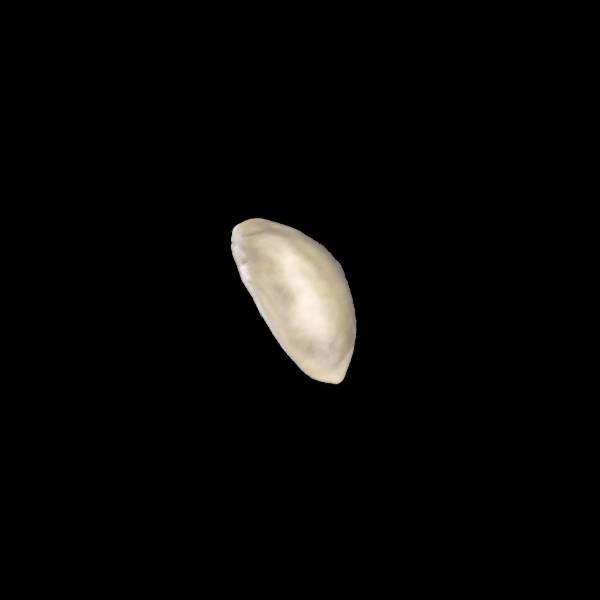} &
         \includegraphics[align=c, width=0.15\linewidth]{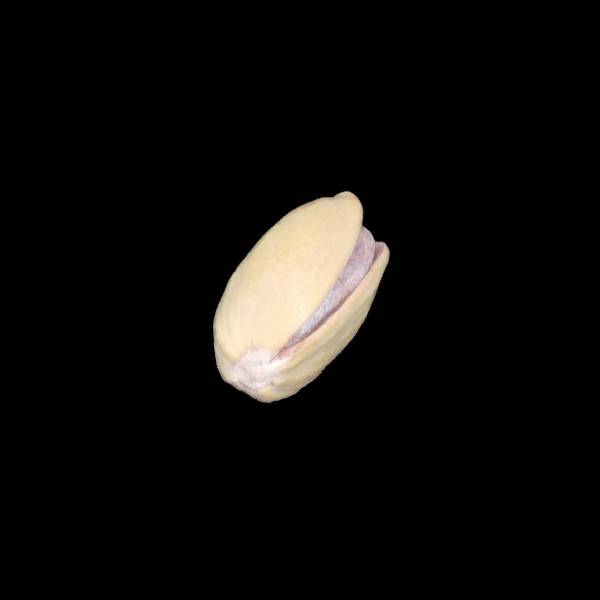}     & \includegraphics[align=c, width=0.15\linewidth]{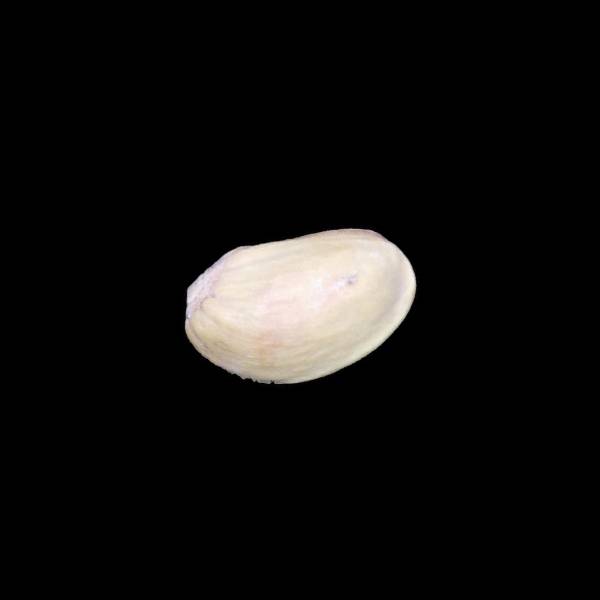} \\
    \end{tabular}
    \caption{Synthetic images using two images \textit{Kirmizi} and \textit{Siirt} pistachios \citep{singh2022}. Each image has a resolution of $600\times 600$.   }
    \label{xy}
\end{figure}
\subsection{Contours extractions and shapes estimation}
From the obtained synthetic images, we extract the contours using the Otsu and the marching squares algorithms \cite{maple2003, otsu}. From the contours, we estimate the Fréchet mean $\boldsymbol{\mu}$, and the shapes $\mathbf{\tilde C}$ using the framework proposed in \cite{axe1}. Figure \ref{cont_all} represents the obtained results using $M_2=30$ Fourier functions, to reconstruct the functional forms of the contours.\par 
\begin{figure}[ht]
    \centering
    \begin{tabular}{ c c c c c c}
    & (a) & (b) & (c) \\
          $\mathbf{C}$ & \includegraphics[width=0.15\linewidth, align=c]{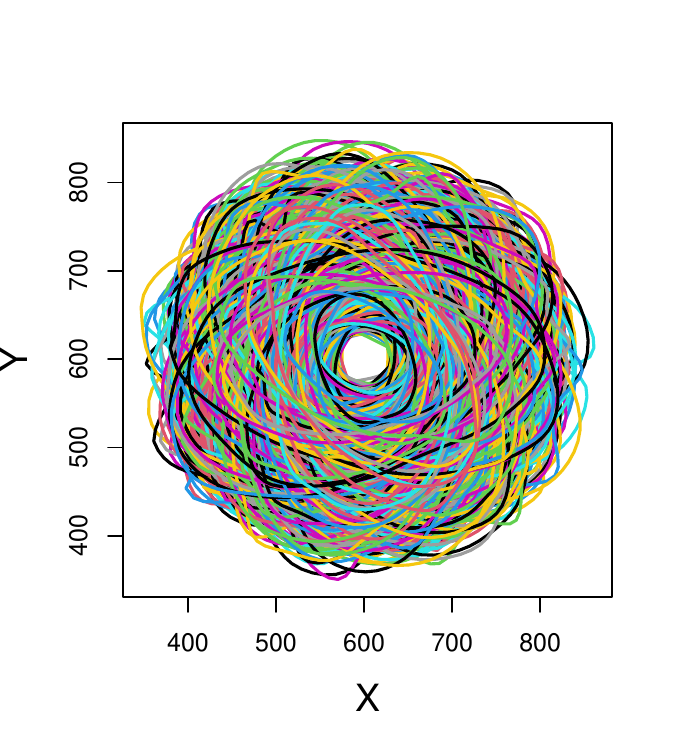}  &      \includegraphics[width=0.15\linewidth, align=c]{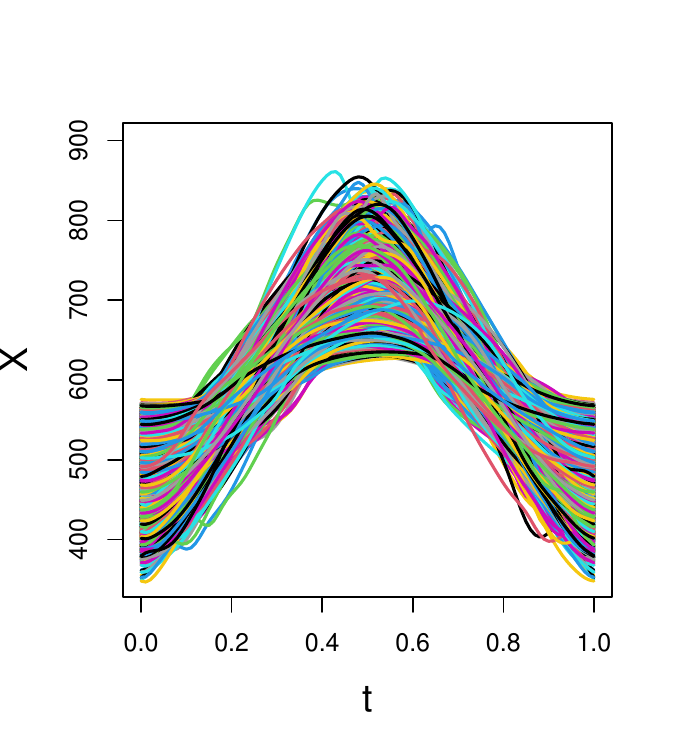} &     \includegraphics[width=0.15\linewidth, align=c]{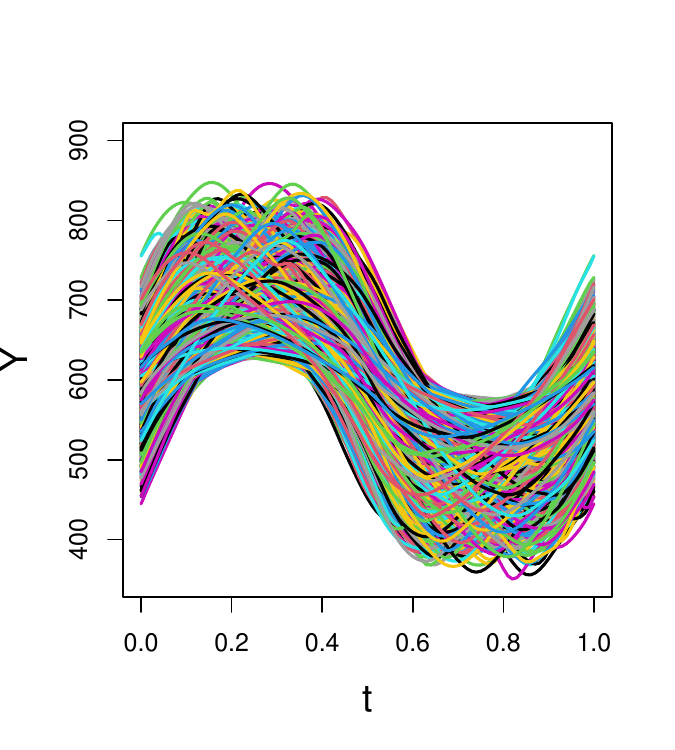} \\
                  $\mathbf{C}^*$ &\includegraphics[width=0.15\linewidth, align=c]{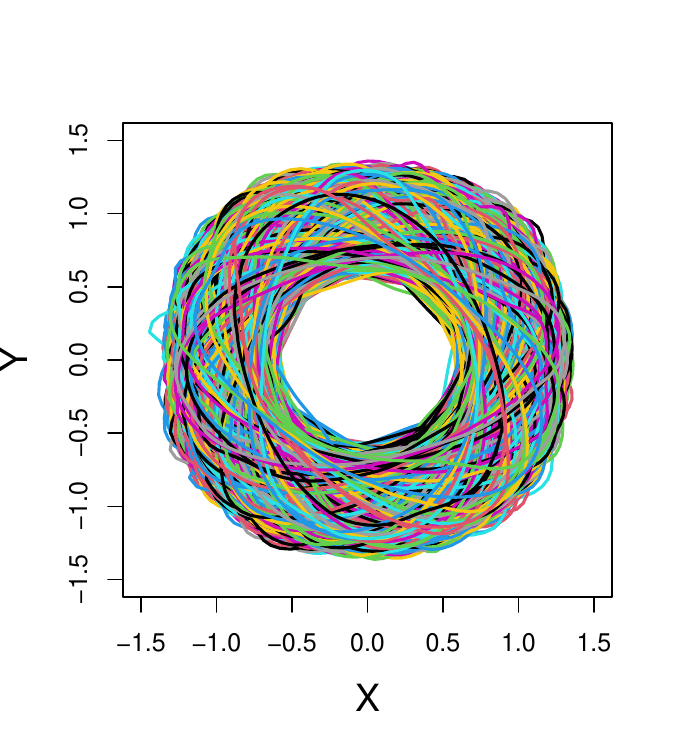}  &    \includegraphics[width=0.15\linewidth, align=c]{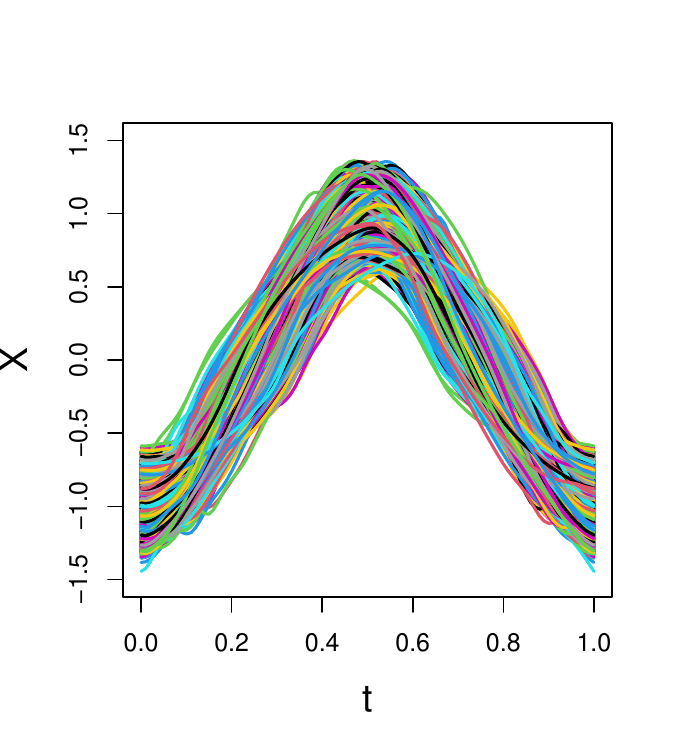} &    \includegraphics[width=0.15\linewidth, align=c]{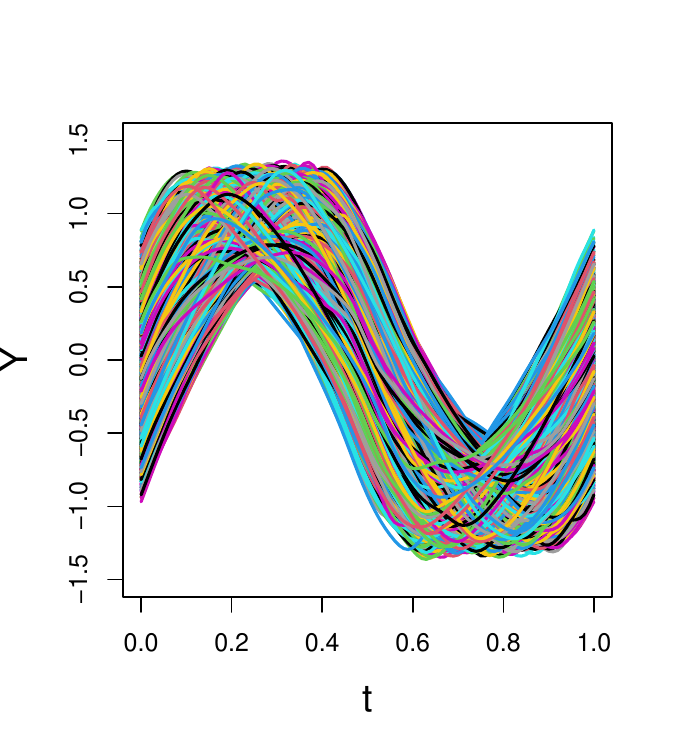} \\ 
$\hat{\boldsymbol{\mu}}$ &\includegraphics[width=0.15\linewidth, align=c]{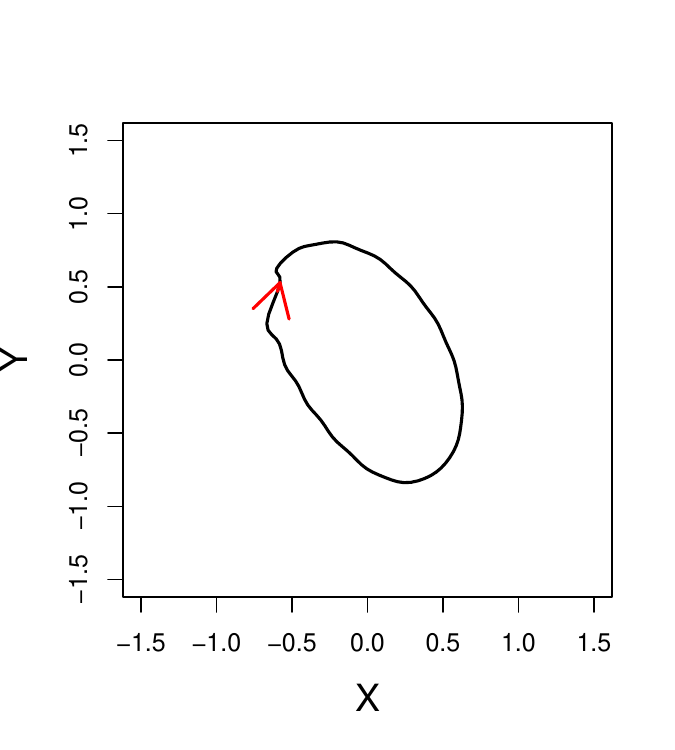}  &    \includegraphics[width=0.15\linewidth, align=c]{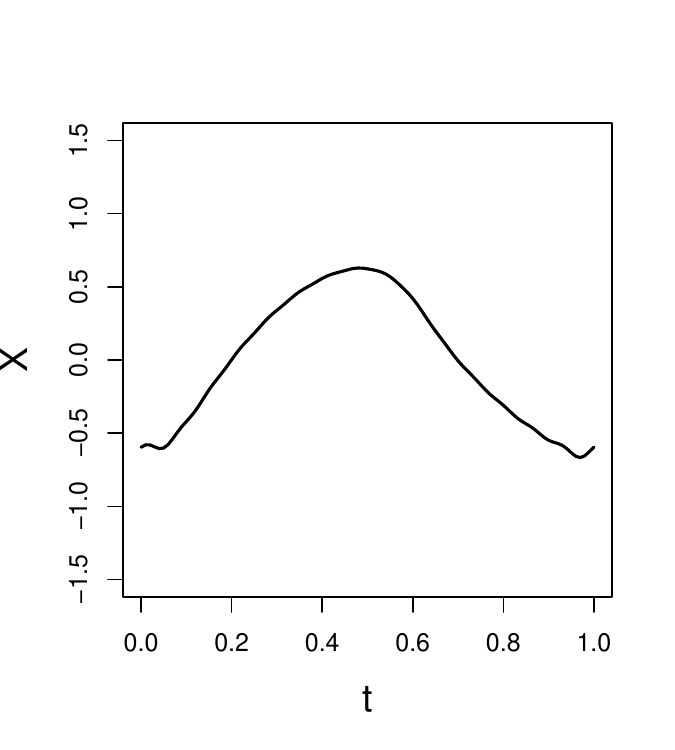} &    \includegraphics[width=0.15\linewidth, align=c]{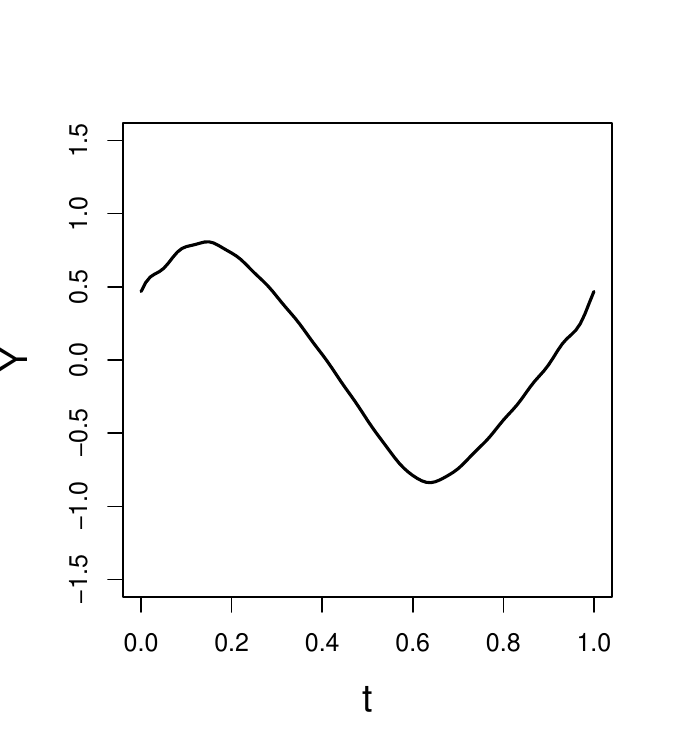} \\
$\mathbf{\tilde C}$& \includegraphics[width=0.15\linewidth, align=c]{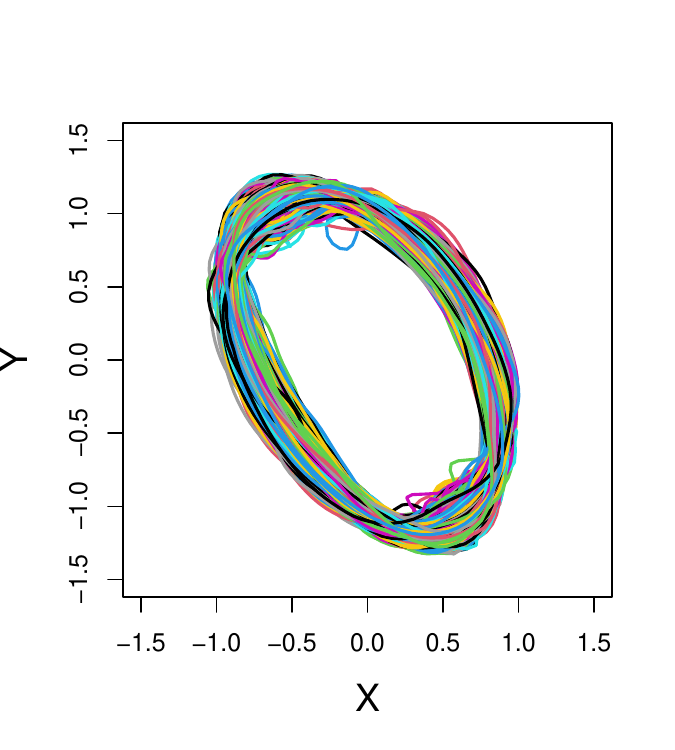} &    \includegraphics[width=0.15\linewidth, align=c]{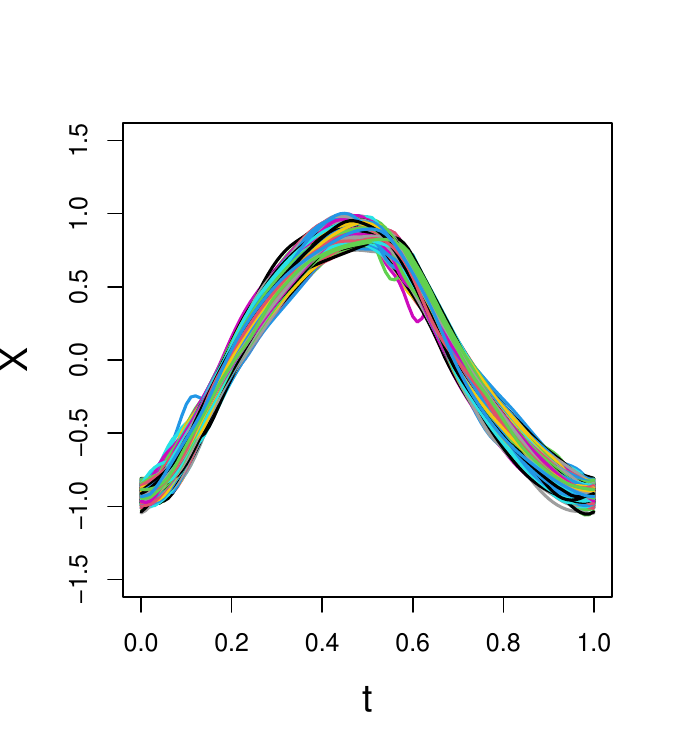} & \includegraphics[width=0.15\linewidth, align=c]{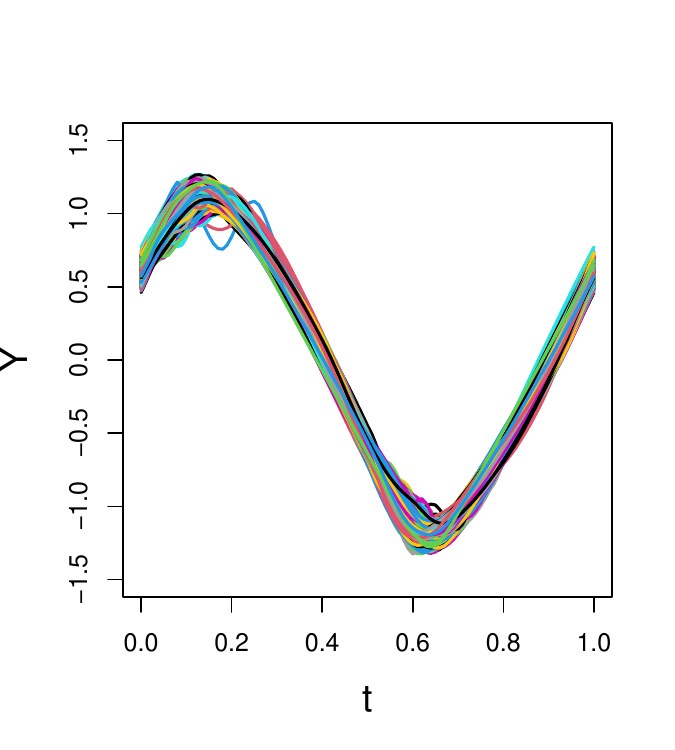}
    \end{tabular}
    \caption{From the contours $\mathbf{C}$, we estimate the preshapes ($\mathbf{C}^*$) by discarding the effect of the translation and scaling. Then, the Frechet mean $\boldsymbol{\mu}$ of the preshape is estimated. The shapes $\mathbf{\tilde C}$ are then estimated using $\boldsymbol{\mu}$. (a) represents the planar curves, (b) and (c) are the coordinates functions, respectively $C_x(t)$ and $C_y(t)$. }
    \label{cont_all}
\end{figure}
It also shows that the deformation variables are well-estimated, as the shape observations $\mathbf{\tilde C}$ are aligned to $\hat{\boldsymbol{\mu}}$, their estimated Fréchet mean.  We recall that the preshapes represent the contours without the global transformations, namely translations and scalings. 
\subsection{Textures}
From the previous part, we obtain the aligned and centered  curves $\mathbf{S}=\rho \mathbf{\tilde C}$ and the translation parameter $\mathbf{c}$ of each contour. With them, we estimate the texture variable of each image using \eqref{text_eq}. Figure \ref{text_ap} presents an example of the texture and the aligned and centered curve for a given image. 
\begin{figure}[ht]
    \centering
    \begin{tabular}{c c c}
         (a) & (b) & (c) \\
         \includegraphics[align=c, width=.2\linewidth]{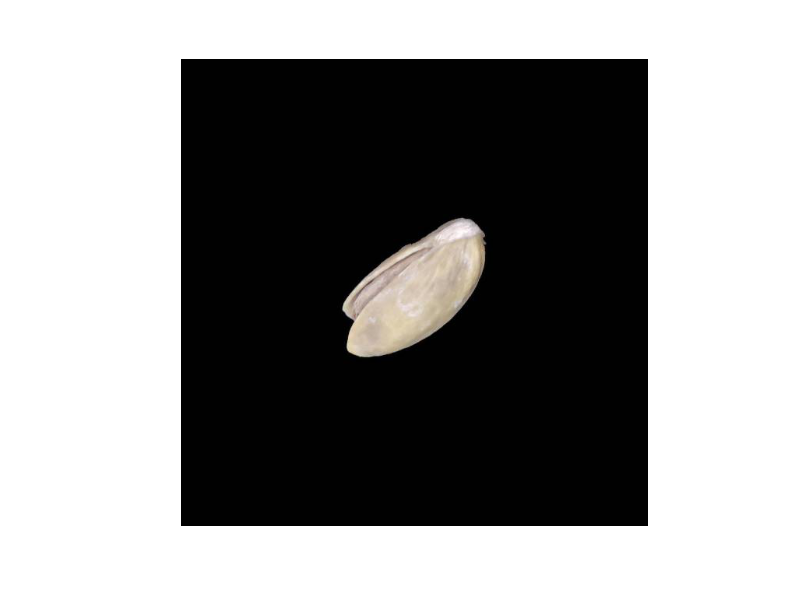}& \includegraphics[align=c, width=.17\linewidth]{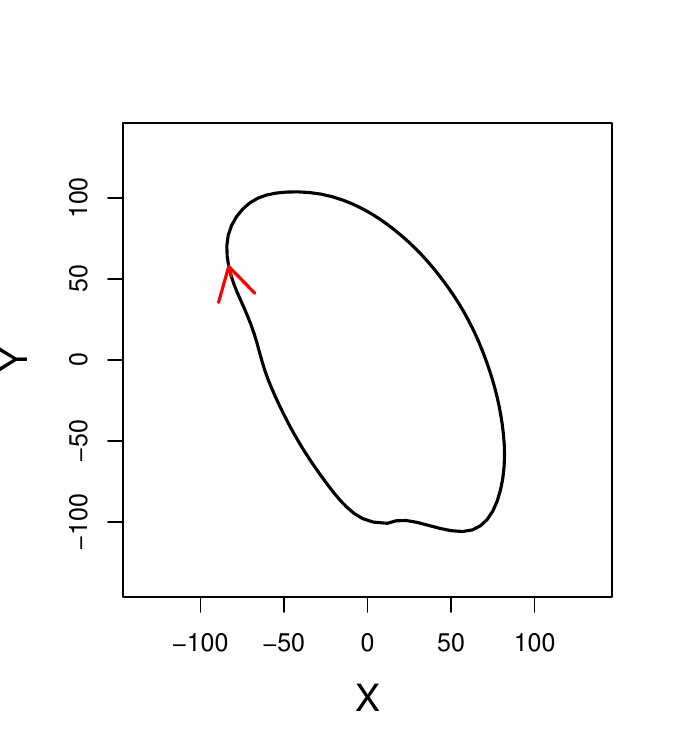} & \includegraphics[align=c, width=.2\linewidth]{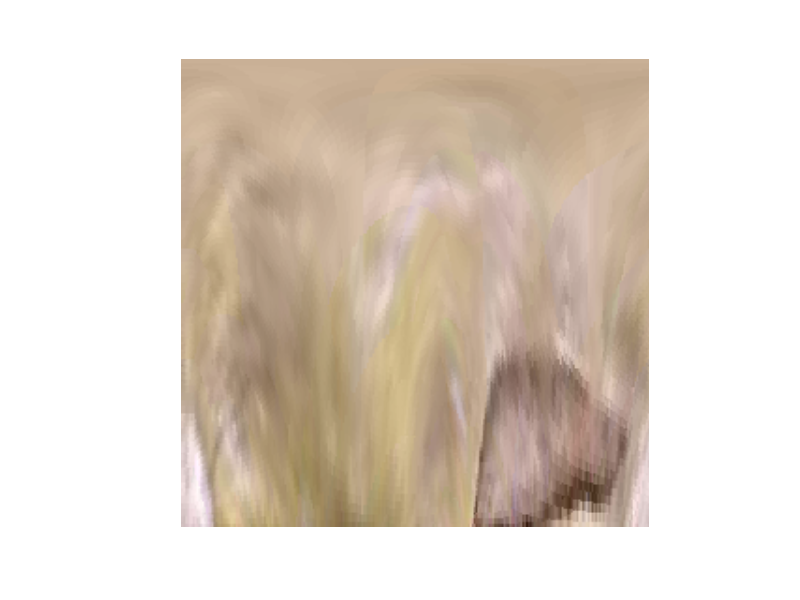}  \\
         & 
    \end{tabular}
    \caption{Original image (a) with its functional components (b \& c). (b) is the associated aligned and centered curve $\mathbf{S}$ and (c) represents the associated texture $\mathbf{T}\circ g$.   }
    \label{text_ap}
\end{figure}
Note that the represented textures are obtained using all the channels of the images. For each image, the operation consists of estimating the texture of each channel. The obtained texture is then a $3$-dimensional array, so a colored texture variable. \par For explanatory purposes, we compute the empirical texture mean for each class, see Figure \ref{txt_mean}. As in Figure~\ref{xx}, it seems that \textit{Kirmizi} and \textit{Siirt} classes of pistachios have a slight difference in their tints. Indeed, Figure \ref{txt_mean} shows that in average, \textit{Siirt} pistachios are lighter than \textit{Kirmizi} ones. This emphasizes the importance to consider textures in the discrimination model.    
\begin{figure}[h]
    \centering
    \begin{tabular}{c c}
         \textit{Kirmizi} &           \textit{Siirt} \\\includegraphics[width=.25\linewidth]{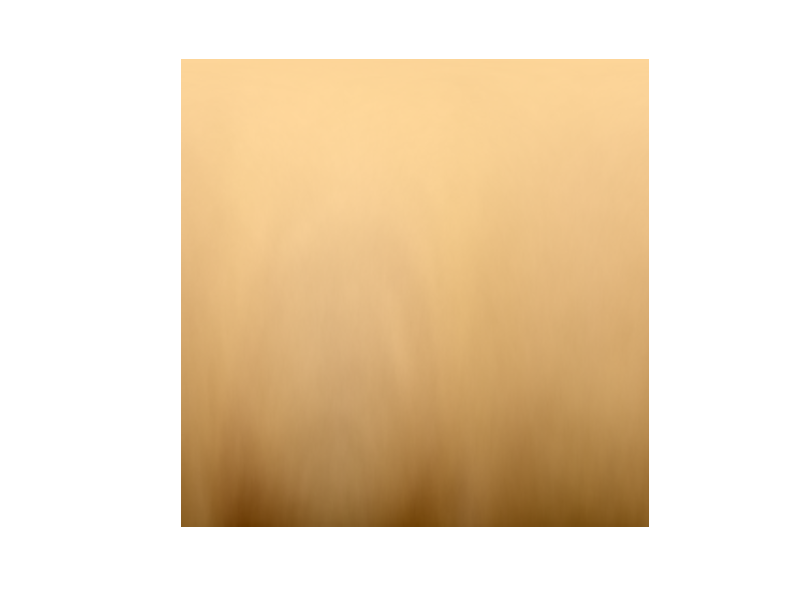}& \includegraphics[width=.25\linewidth]{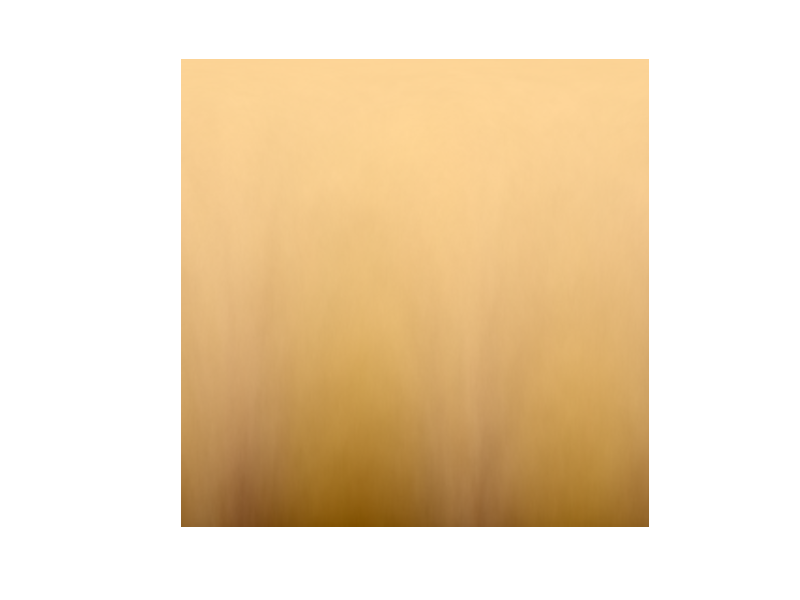} 
    \end{tabular}
    \caption{Empirical texture means by class.}
    \label{txt_mean}
\end{figure}
\subsection{Results}
To evaluate the proposed representation of images, we train linear discriminant analysis methods based on our framework, using MFPCA and MFPLS on the synthetic images. In both, we use $10$ bi-dimensional splines for each channel of textures and $10$ quadratic splines for the observations of $\mathcal{T}_{\boldsymbol{\mu}}(\mathbf{\tilde C})$. Moreover, for comparison purposes, we also estimate these methods using solely the input images as proposed in  \citep{image-happ}. This allows to compare the two models with different frameworks, our and the scalar-on-image regression framework. In this latter setting, the image is seen as three random surfaces (corresponding to each of its channel), and we use $10$ bi-dimensional splines on each channel for reconstructing its functional form.\par 
For each method, the table \ref{res} presents the performance (accuracy and F-1 score) obtained using $10$-cross-validation of our proposed framework (OUR) and the competitor framework (COMP), which sees the image as random surfaces. 

\begin{table}[ht]
\centering
\begin{tabular}{ccccc}
  \hline
 & & \textbf{MFPCA} & \textbf{MFPLS} \\ 
\multirow{ 2}{*}{OUR}& ACC & 0.77 & 0.76 \\ 
  & F1 & 0.74 & 0.73 \\ 
  \hline 
  \multirow{ 2}{*}{COMP}&ACC & 0.63 & 0.57 \\ 
  &F1 & 0.42 & 0.54 \\ 
\end{tabular}
\caption{The results obtained using $10$ cross-validation }
\label{res}
\end{table}

The results show that the proposed representation outperforms the scalar‑on‑image baseline. Indeed, models estimated using the competitor formalism have accuracies of about $60\%$ and F-1 of $50\%$, which is equivalent to a random baseline model predictions. 
\section{Discussion and perspectives}
\label{disc}
This work proposes a new framework for statistical analyses of images. It extends the work of \cite{axe1} by taking into account the color variations in images. For doing so, we introduce and study a mapping of the interior of the main object of study to a more convenient space. This mapping is then applied in a supervised classification framework, where the obtained results demonstrate its appealingness. \par 
Our work proposes a new paradigm for image, by explicitly taking the main aspect of images, and so differs of recent works in statistical image analysis (see e.g \cite{image-2}, \cite{image-happ}); which aim to extend the pixel-based approach to a continuous setting. As the motivation application shows such strategy can be limited when images have substantial deformations, especially when using linear models. Moreover, the proposed framework is a more frugal representation of images: it is a sparse representation of shapes, and only considers the color variation of interest in images. Although, we focused on linear models in the statistical learning, non-linear models can also be explored, for example following the recent works of \cite{wu2024, wang2024}.  \par 
This work heavily relies on the assumption that the object’s interior is a star‑shaped domain. This assumption allows having a mapping from the interior of the object to $[0, 1]^2 \subset \mathbb{R}^2$. A natural extension of this work is the use of more general interior domains, for example by exploring the smallest star-domain which contains the interior, in a convex-hull fashion. However, such extension would necessarily lead to consider the background colors. Future works should focus on this aspect, while studying the incorporation of several objects in the analysis.

\section*{Acknowledgements} \label{ackn}

The author would like to thank Prof.Descary and Prof.Beaulac for their insightfully suggestions and comments.

\bibliographystyle{apalike}
\bibliography{ref}
\newpage 
\appendix
\section{Technical details}
\label{annex}
\begin{proof}[Proof of Proposition \ref{prop1}]
Using the fact that $\textbf{C}(t)\in \mathcal{D}$, $t\in [0, 1]$, we have that \eqref{star} is equivalent to  
\begin{align}
    \{(1-\lambda)\mathbf{c} + \lambda \mathbf{C}(t)\mid \lambda \in\left[0,1\right]\}\subset \mathcal{D} \iff   \{\mathbf{c} + \mathbf{O}\left(\lambda \mathbf{S}(t)\right) \mid \lambda \in\left[0,1\right]\}\subset \mathcal{D}.
    \label{int}
\end{align}
Remark that $v=\lambda e^{i2\pi t}$ is a global parametrization of $\mathbf{D}$, then $$\varphi(v)= \lambda \mathbf{S}(t)= |v| \mathbf{S}\left(\frac{\text{arg}(v)}{2\pi}\right).$$ This concludes the proof.\end{proof}
\begin{proof}[Proof of Proposition \ref{prop2}] \it First, we state the proprieties of $\varphi(v)$: $\varphi(\lambda v)= \lambda \varphi(v),\ $ $\lambda\geq 0$ ;  $\varphi(e^{i\theta})= \mathbf{S}\left(\frac{\theta }{2\pi} \right)$, $\theta\in [0, 2\pi]$, which are derived directly from the definition of $\varphi(v)$.  \par 
Now, we prove Proposition \ref{prop2} by relying on the two following points: 
\begin{enumerate}
    \item \textbf{Existence}: Since $\mathcal{D}$ is a star-shaped domain and $\mathbf{S}(t)$ is the frontier of the domain, $u\in \mathcal{D}$, implies that there exists $\theta^*$ and $\lambda \in (0, 1]$ such as $$
u=\lambda \mathbf{S}\left(\frac{\theta^*}{2\pi}\right) +\mathbf{c}. 
$$
this expression is equivalent to $$
  u=\mathbf{O}\varphi(v)+\mathbf{c}
$$
where $$v=\frac{\norm{u-\mathbf{c}}_2}{\norm{\mathbf{S}\left(\frac{\theta^*}{2\pi}\right) }_2} e^{i\theta^*} \ \text{and }  \alpha(\theta^*)= \arg\left( \begin{pmatrix}
    1 & i
\end{pmatrix} (u-\mathbf{c}) \right). $$
\item \textbf{Uniqueness}: Let assume that there exist $v_1$ and $v_2$, such as 
\begin{align*}
    u&= \mathbf{O}\varphi(v_1)+ \mathbf{c} \\ 
    u&= \mathbf{O}\varphi(v_2)+ \mathbf{c},
\end{align*}
if we note $v_1=\lambda_1 e^{i\theta_1}$ and $v_2=\lambda_2 e^{i\theta_2}$, we then have that $$
\arg( \begin{pmatrix}
    1 & i 
\end{pmatrix} (\mathbf{O}^\top u-\mathbf{c}))= \alpha(\theta_1)=\alpha(\theta_2).  
$$
Since $\alpha(\theta)$ is injective, this implies that $\theta_1=\theta_2$. Then, 
\begin{align*}
    u&=\mathbf{O}\lambda_1 \varphi(e^{i\theta_1})+\mathbf{c}\\ 
    u&=\mathbf{O}\lambda_2 \varphi(e^{i\theta_1})+\mathbf{c}
\end{align*}
by identification, we have that $\lambda_1=\lambda_2$. Finally, $v_1=v_2$. 
\end{enumerate}
\end{proof}
\end{document}